\documentclass{article}
\usepackage{ijcai26}

\usepackage{times}
\usepackage{soul}
\usepackage{url}
\usepackage[hidelinks]{hyperref}
\usepackage[utf8]{inputenc}
\usepackage[small]{caption}
\usepackage{graphicx}
\usepackage{amsmath}
\usepackage{amsthm}
\usepackage{comment}
\usepackage{bm}
\usepackage{amssymb}
\usepackage{booktabs}
\usepackage{algorithm}
\usepackage{algorithmic}
\usepackage[switch]{lineno}

\newtheorem{theorem}{Theorem}
\newtheorem{proposition}[theorem]{Proposition}

\title{Cost of Structural Learning under Censored Feedback: \\ A Threshold-Bandit Approach}

\author{
Michael Ledford
\And
William Regli\\
\affiliations
University of Maryland, College Park\\
\emails
\{mledfor, regli\}@umd.edu
}

\begin{document}

\maketitle

\begin{abstract} 

In many multi-agent applications, tasks yield rewards only when executed by a coalition meeting an unknown size threshold; otherwise, feedback is fully censored. This censorship creates an identifiability problem: agents cannot distinguish stochastic failure from insufficient coordination. We formalize this setting as the Threshold-Activated Cooperative Multi-Armed Bandit (TAC-MAB) and analyze it under both centralized and decentralized coordination. We show that a centralized algorithm (C-TAC) achieves cumulative regret $O(\log T)$, decomposed into a structural-search term that captures the cost of resolving feasibility under censored feedback and a statistical-monitoring term for value estimation. We then introduce \textbf{D-TAC}, a decentralized event-triggered protocol in which agents synchronize only when their structural beliefs change. Empirically, D-TAC achieves a $23{\times}$ reduction in communication relative to the centralized baseline while preserving feasibility alignment under conservative belief fusion. These results characterize the coordination cost of learning under censored feedback and show that near-centralized communication efficiency is achievable without continuous synchronization.

\end{abstract}

\section{Introduction}
\label{sec:intro}

Cooperative multi-agent systems often operate in environments where task success depends on coordinated action by multiple agents, and execution feedback becomes informative only when a sufficient coalition acts jointly. Consider domains such as search-and-rescue~\cite{cao2024hma}, where multiple agents explore an environment and success depends on deploying a coalition of the right size acting in unison. Similar coordination dependencies arise in distributed sensing, signal jamming, and logistics coordination, where agents must combine capabilities to achieve an outcome.

In these settings, the coalition size required to produce informative feedback can be viewed as a \emph{latent feasibility threshold}, a quantitative constraint that is hidden from agents but determines whether a joint action yields any reward. When too few agents are deployed, execution fails silently without producing informative feedback, making stochastic failure indistinguishable from insufficient coordination.

Existing models for cooperative decision making typically assume execution feedback is informative whenever agents act, enabling learning through independent or weakly coordinated exploration~\cite{landgren2021distributed}. This challenge persists across modeling frameworks: MARL methods under sparse rewards~\cite{mahajan2019maven} treat feasibility as a learning problem but do not exploit its structural nature, while Dec-POMDP formulations~\cite{oliehoek2016concise} assume observation models that censored feedback violates. A multi-armed bandit framing, by contrast, isolates the allocation decision from state and policy dynamics, making the coordination cost itself analytically tractable. When task outcomes are gated by unknown coalition-size requirements, execution attempts below the feasibility threshold yield fully censored feedback, and the probability of independently forming a feasible coalition vanishes as coalition-size requirements grow. Independent exploration thus cannot reliably resolve feasibility on its own.

Coordination requirements could in principle be specified offline, but feasibility often depends on transient or adversarial conditions that resist pre-programming, and offline training suffers from the sim-to-real gap~\cite{tobin2017domain}. Teams must therefore infer feasibility online. Yet committing more agents is not free: practical deployments incur communication and execution costs~\cite{chakraborty2017coordinated} (energy, bandwidth, time) when allocating larger coalitions. This creates an exploration dilemma: independent learning fails under censored feedback, while naive full coordination is prohibitively costly. We measure the resulting inefficiency using regret, which captures the cost of operating under incorrect feasibility beliefs.

To address this challenge, we introduce the \textbf{Threshold-Activated Cooperative Multi-Armed Bandit (TAC-MAB)}, a model that isolates the structural difficulty of feasibility-gated feedback. We first analyze an idealized centralized baseline (C-TAC) and prove $O(\log T)$ cumulative regret, decomposed into a structural-search term that captures the cost of resolving feasibility under censored feedback and a statistical-monitoring term for value estimation (Theorem~\ref{thm:regret}).

We then propose \textbf{D-TAC}, a decentralized event-triggered protocol in which agents synchronize only when structural beliefs change. Once structural beliefs align, D-TAC executes the same plan as C-TAC and empirically approaches its asymptotic performance, modulo a transient during the structural-learning phase (Section~\ref{sec:dtac}). Empirically, D-TAC achieves a $23{\times}$ reduction in communication relative to the centralized baseline while preserving feasibility alignment through conservative belief fusion.

\section{Problem Formulation}
\label{sec:problem}

We study a cooperative multi-agent learning problem in which $M$ agents must allocate themselves across $K$ tasks, where each task requires an unknown integer number of agents to activate. Below this threshold, execution attempts yield no informative feedback. We refer to this setting as the \textbf{Threshold-Activated Cooperative Multi-Armed Bandit (TAC-MAB)}.

TAC-MAB is a multi-agent specialization of the censored semi-bandit framework of Verma~\textit{et~al.}~\shortcite{verma2019censored}, which studies single-learner allocation with continuous-valued thresholds and a divisible resource budget. Our setting differs in three structural ways: the resource is a coalition of $M$ distinct agents (rather than a divisible budget), thresholds are integer-valued (rather than continuous), and observations are partitioned by agent assignment under decentralized execution. These differences enable the study of decentralized coordination protocols, which we develop in Section~\ref{sec:dtac}.

\subsection{Agents and Task Types}
\label{sec:agents-tasks}

The team consists of $M$ homogeneous agents indexed by $i \in \{1,\ldots,M\}$ operating over a finite horizon $T$. The environment contains $K$ stationary tasks indexed by $k \in \{1,\ldots,K\}$, each characterized by a feasibility threshold $\tau_k \in \mathbb{Z}_{\ge 1}$ (the minimum coalition size required to activate task $k$, possibly with $\tau_k > M$ for structurally infeasible tasks), a success probability $p_k \in [0,1]$ (the probability that task $k$ succeeds when executed by a feasible coalition), and a task value $v_k \in \mathbb{R}_{+}$ (the reward obtained on a successful execution). These parameters are stationary over the horizon. Together with the fixed feasibility set $\mathcal{C}$ defined in Section~\ref{sec:objective-regret}, this makes the oracle's optimal allocation $\boldsymbol{c}^\star$ time-invariant.

An instance of TAC-MAB is identified by $(\boldsymbol{\tau}, \boldsymbol{p}, \boldsymbol{v}, M, T)$. The vectors $\boldsymbol{\tau}$ and $\boldsymbol{p}$ are unknown and must be learned through interaction. The remaining quantities are known to the team: $K$, $M$, $T$, the per-task values $\boldsymbol{v} = \{v_k\}_{k=1}^K$, and a known lower bound $p_{\min} \in (0, 1]$ such that $p_k \geq p_{\min}$ for all feasible tasks (i.e., tasks with $\tau_k \leq M$). We treat $\boldsymbol{v}$ as known; this is standard in resource-allocation problems with pre-specified task values~\cite{verma2019censored} and lets us isolate the cost of learning $(\boldsymbol{\tau}, \boldsymbol{p})$. Extending to unknown $\boldsymbol{v}$ is left to future work.  

We assume $K > M$ throughout, the regime in which the team cannot evaluate all tasks in parallel and must select a subset to attempt each round. This is what makes the allocation problem non-trivial: even with full knowledge of $(\boldsymbol{\tau}, \boldsymbol{p}, \boldsymbol{v})$, the optimal allocation solves a 0/1 knapsack with weights $\tau_k$ and values $p_k v_k$.

\subsection{Joint Actions and Censored Feedback}
\label{sec:censored-feedback}

In round $t$, the team selects $\boldsymbol{c}_t = (c_{1,t},\ldots,c_{K,t})$ with $c_{k,t} \in \{0,\ldots,M\}$ and $\sum_k c_{k,t} \leq M$. For each task $k$ and round $t$, $X_{k,t} \sim \mathrm{Bernoulli}(p_k)$ is drawn independently across tasks and rounds. The observed outcome is
\begin{equation}
Y_{k,t} =
\begin{cases}
X_{k,t}\, v_k, & \text{if } c_{k,t} \geq \tau_k, \\
0, & \text{if } c_{k,t} < \tau_k
\end{cases}
\end{equation}

\paragraph{Censoring at observation, not generation.}
$X_{k,t}$ is realized regardless of coalition size, but the team observes $Y_{k,t} = 0$ deterministically when $c_{k,t} < \tau_k$. Crucially, when $Y_{k,t} = 0$, the team cannot distinguish censoring ($c_{k,t} < \tau_k$) from stochastic failure ($X_{k,t} = 0$ with $c_{k,t} \geq \tau_k$). This is the identifiability challenge that motivates our analysis.

\paragraph{Observation convention.}
We treat $Y_{k,t}$ as the per-task aggregate outcome; this is what enters the regret definition. In our implementation, the reward $v_k$ is distributed evenly across the coalition assigned to task $k$; each agent on task $k$ thus observes the success indicator $\mathbf{1}\{Y_{k,t} > 0\}$ via its own share. The team-level cumulative reward and our regret analysis are invariant to this distribution choice.

\subsection{Observations and Communication}
\label{sec:observations}

Each agent $i$ observes only the outcome $Y_{k,t}$ for its assigned task at round $t$, via the share convention above. Agents do not observe outcomes of other tasks, nor the thresholds $\tau_k$, but each agent knows the coalition size assigned to its own task (via the coordinator in the centralized setting and via the protocol in the decentralized setting). A successful outcome ($Y_{k,t} > 0$) reveals that the executed coalition was feasible and that $X_{k,t} = 1$; a failed outcome ($Y_{k,t} = 0$) is uninformative about which of the two failure modes occurred.

Agents may communicate at the end of each round; communication, when it occurs, is reliable and synchronous within the round. We abstract communication cost as the number of messages exchanged per round (independent of payload size), corresponding to a shared-medium broadcast model for one-to-many transmissions and point-to-point unicast for direct messages. A single broadcast counts as one message, while all-to-all exchange among $M$ agents incurs $\mathcal{O}(M^2)$ messages.

We analyze TAC-MAB under two coordination architectures. In the \emph{centralized} architecture, all observations in a round are available to a coordinator that maintains global estimates and selects joint actions; this serves as a theoretical baseline (Section~\ref{sec:ctac}). In the \emph{decentralized} architecture, each agent has access only to its own observations and may exchange messages with peers subject to the cost model above; coordination protocols are developed in Section~\ref{sec:dtac}. In our centralized analysis (Section~\ref{sec:ctac}), we treat coordinator-to-agent communication as a zero-cost idealization to isolate the statistical cost of learning; decentralized communication costs are accounted for in Section~\ref{sec:dtac}.

\subsection{Objective and Regret}
\label{sec:objective-regret}

The objective of the team is to maximize cumulative reward over the horizon $T$. This requires solving two coupled sub-problems online:
\begin{enumerate}
    \item \textbf{Threshold learning under censored feedback.} Estimate $\boldsymbol{\tau}$ from observed outcomes, where failures below $\tau_k$ provide no information distinguishing insufficient coordination from stochastic failure.
    \item \textbf{Optimal allocation under resource constraints.} Given threshold estimates $\hat{\boldsymbol{\tau}}$ and value estimates $\hat{\boldsymbol{\mu}}$, select an allocation $\boldsymbol{c} \in \mathcal{C}$ that maximizes expected reward, where $\mathcal{C} = \{\boldsymbol{c} \in \mathbb{Z}^K_{\geq 0} : \sum_k c_k \leq M\}$. Since $K > M$, this is a 0/1 knapsack.
\end{enumerate}
These sub-problems are coupled through the per-round resource constraint: probing a task's threshold requires committing agents to it, which forecloses other allocations that round. Threshold learning therefore competes with reward collection for the same scarce resource.

Let $\mu_k = p_k v_k$. The oracle benchmark knows $(\boldsymbol{\tau}, \boldsymbol{\mu})$ and selects
\begin{equation}
\boldsymbol{c}^\star \in \arg\max_{\boldsymbol{c} \in \mathcal{C}} \sum_{k=1}^{K} \mathbf{1}\{c_k \geq \tau_k\}\, \mu_k,
\end{equation}
which does not depend on $t$ by stationarity. We assume a unique maximizer for clarity. Let
\[
\mu^\star \;=\; \sum_{k=1}^{K} \mathbf{1}\{c^\star_k \geq \tau_k\}\, \mu_k
\]
denote the expected reward of this optimal joint allocation. Cumulative regret is
\begin{equation}
R(T) = \mathbb{E}\!\left[\sum_{t=1}^{T} \left( \mu^\star - \sum_{k=1}^{K} \mathbf{1}\{c_{k,t} \geq \tau_k\}\, X_{k,t}\, v_k \right) \right],
\label{eq:regret}
\end{equation}
where the expectation is over $\{X_{k,t}\}$ and any randomness in the learning policy. The oracle is assumed to solve the allocation optimization exactly; regret is measured relative to this ideal benchmark.

Regret in TAC-MAB decomposes into two sources: the cost of resolving threshold uncertainty under censored feedback, and the cost of estimating the success probabilities $\boldsymbol{p}$ under the integer knapsack constraint. Theorem~\ref{thm:regret} (Section~\ref{sec:ctac}) characterizes both.

\section{Centralized Coordination Baseline}
\label{sec:ctac}

We first analyze TAC-MAB under an idealized centralized coordination architecture in which a coordinator observes all task outcomes immediately and broadcasts joint allocations at zero cost. This baseline isolates the statistical difficulty of learning under censored feedback from the logistical cost of coordination, and provides a regret reference for the decentralized strategies in Section~\ref{sec:dtac}.

\subsection{Structural and Statistical Uncertainty}
\label{sec:uncertainty}

Even with perfect information sharing, learning under TAC-MAB couples two sources of uncertainty. \emph{Structural uncertainty} arises from unknown feasibility thresholds $\boldsymbol{\tau}$ that gate informative feedback: when a coalition is below threshold, outcomes are fully censored and failures do not distinguish insufficient coordination from stochastic effects. \emph{Statistical uncertainty} arises from unknown success probabilities $\boldsymbol{p}$ once a task is executed by a feasible coalition. Theorem~\ref{thm:regret} decomposes regret along this split.

To ensure regret is measured only with respect to information limitations, we assume the coordinator has access to an exact planning oracle that solves the integer 0/1 knapsack at each round given current estimates. Under this assumption, regret is attributable to censored feedback and statistical estimation, not computational approximation. For our integer-threshold setting, the knapsack is solved exactly in $O(KM)$ time via dynamic programming.

\subsection{A Constructive Centralized Strategy}
\label{sec:ctac-algo}

We instantiate this baseline as \textbf{C-TAC} (Centralized Threshold-Activated Coordination), shown in Algorithm~\ref{alg:ctac}. We assume the \emph{identifiability condition}: every feasible task $k$ (i.e., $\tau_k \leq M$) satisfies $p_k \ge p_{\min} > 0$ for some known constant $p_{\min}$. Without this condition, feasibility cannot be distinguished from stochastic failure under censored feedback, making sublinear regret impossible.

C-TAC maintains a per-task phase $\phi_k \in \{\textsc{search}, \textsc{monitor}, \textsc{infeasible}\}$, with the first two labeled \emph{active}. A task begins in \textsc{search} with $\hat\tau_k = 1$ and advances $\hat\tau_k$ on repeated failures; the first non-zero observation at $\hat\tau_k$ transitions the task to \textsc{monitor}, freezing $\hat\tau_k$ and confining further updates to the mean-reward estimate $\hat\mu_k$. This avoids spurious pruning once feasibility is confirmed. If $\hat\tau_k$ exceeds $M$ during \textsc{search}, the task is marked \textsc{infeasible} and excluded from planning. Each round, the coordinator solves an exact 0/1 knapsack over active tasks using UCB1-based indices~\cite{auer2002finite} (Algorithm~\ref{alg:ctac}, line 3). We use linear search for threshold updates, appropriate for the team sizes we consider; larger teams could employ doubling or binary search for a $\log M$ reduction at the cost of harder analysis.

\begin{algorithm}[t]
\caption{C-TAC (Centralized Threshold-Activated Coordination)}
\label{alg:ctac}
\begin{algorithmic}[1]
\REQUIRE Tasks $K$, Agents $M$, Horizon $T$, Values $\boldsymbol{v}$, Identifiability bound $p_{\min}$, Failure budget $N_{\max}$, UCB constant $c$
\STATE \textbf{Initialize:} For each task $k$, set $\hat\tau_k \leftarrow 1$, $\phi_k \leftarrow \textsc{search}$, $N_{\text{fail},k} \leftarrow 0$, $\hat\mu_k \leftarrow 0$, $N_k \leftarrow 0$
\FOR{$t = 1, \dots, T$}
    \STATE \textbf{Planning:} $\boldsymbol{c}_t^* \leftarrow$ exact 0/1 knapsack over active tasks $\{k : \phi_k \neq \textsc{infeasible}\}$ with weights $\hat\tau_k$ and per-task indices $\hat\mu_k + v_k \sqrt{c \log t / N_k}$ ($N_k = 0$ treated as $+\infty$)
    \STATE \textbf{Execution:} Execute $\boldsymbol{c}_t^*$, observe outcomes $\{Y_{k,t}\}$
    \FOR{each task $k$ with coalition size $c_{k,t}^* \ge \hat\tau_k$}
        \IF{$\phi_k = \textsc{search}$}
            \IF{$Y_{k,t} > 0$}
                \STATE $\phi_k \leftarrow \textsc{monitor}$; $N_{\text{fail},k} \leftarrow 0$ \COMMENT{feasibility confirmed}
                \STATE $N_k \leftarrow N_k + 1$; update $\hat\mu_k$ via empirical mean
            \ELSE
                \STATE $N_{\text{fail},k} \leftarrow N_{\text{fail},k} + 1$ \COMMENT{censored failure}
                \IF{$N_{\text{fail},k} \ge N_{\max}$}
                    \STATE $\hat\tau_k \leftarrow \hat\tau_k + 1$; $N_{\text{fail},k} \leftarrow 0$ \COMMENT{linear prune}
                    \IF{$\hat\tau_k > M$}
                        \STATE $\phi_k \leftarrow \textsc{infeasible}$
                    \ENDIF
                \ENDIF
            \ENDIF
        \ELSIF{$\phi_k = \textsc{monitor}$}
            \STATE $N_k \leftarrow N_k + 1$; update $\hat\mu_k$ via empirical mean
        \ENDIF
    \ENDFOR
\ENDFOR
\end{algorithmic}
\end{algorithm}

Before stating the theoretical guarantees, we formally define the suboptimality gaps for our combinatorial setting, adapting standard definitions from the combinatorial semi-bandit literature~\cite{chen2013combinatorial,combes2015combinatorial}. Let $\mu(\boldsymbol{c}) = \sum_{j=1}^{K} \mathbf{1}\{c_j \geq \tau_j\}\, \mu_j$ denote the expected reward of any feasible allocation $\boldsymbol{c} \in \mathcal{C}$, and let $\Delta_{\boldsymbol{c}} = \mu^\star - \mu(\boldsymbol{c})$ be its corresponding suboptimality gap. For each task $k$, the task-specific gap $\Delta_k$ is defined as the minimum positive suboptimality gap among all allocations that feasibly activate task $k$:
\begin{equation}
\Delta_k = \min_{\boldsymbol{c} \in \mathcal{C} : \, c_k \geq \tau_k, \, \Delta_{\boldsymbol{c}} > 0} \Delta_{\boldsymbol{c}}
\end{equation}
Furthermore, let $\Delta_{\max} = \max_{\boldsymbol{c} \in \mathcal{C}} \Delta_{\boldsymbol{c}}$ denote the maximum possible suboptimality gap across all valid allocations.

\begin{theorem}[Centralized TAC-MAB Regret]
\label{thm:regret}
Consider the TAC-MAB problem with $K$ tasks and $M$ agents over a time horizon $T$, under centralized coordination. Setting the failure budget to $N_{\max} = \lceil \frac{2 \log T}{\log(1/(1-p_{\min}))} \rceil$, the expected cumulative team regret $R(T)$ of C-TAC satisfies:
\begin{multline}
R(T) = \mathcal{O}\Bigg(
\underbrace{\sum_{k=1}^K \frac{\min(\tau_k, M) \log T}{p_{\min}} \mu^\star}_{\text{structural search}} \,+ \\
\underbrace{\sum_{k: \tau_k \leq M} \frac{v_k^2 \log T}{\Delta_k}}_{\text{statistical monitoring}} \,+ 
\underbrace{K \Delta_{\max}}_{\text{tail failures}}
\Bigg)
\end{multline}
\end{theorem}

\begin{proof}[Proof sketch]
We decompose the expected cumulative regret into three components: structural search under censored feedback, statistical monitoring of feasible allocations, and a constant bounding the tail failures.

\emph{(i) Structural search.} For any feasible coalition allocation $c_{k,t} \geq \tau_k$, the probability of $N_{\max}$ consecutive Bernoulli$(p_k)$ failures is at most $(1-p_k)^{N_{\max}}$, which by the identifiability condition ($p_k \geq p_{\min}$) is at most $(1-p_{\min})^{N_{\max}} \leq 1/T^2$. By a union bound over $T$ rounds and $K$ tasks, with probability at least $1 - K/T$ no feasible task is incorrectly marked \textsc{infeasible}. The \textsc{search}-only pruning rule advances $\hat\tau_k$ incrementally, taking at most $\min(\tau_k, M)$ phases. Each phase costs $N_{\max} = O(\log T / p_{\min})$ rounds, during which the team forgoes a worst-case expected reward of up to $\mu^\star$ (the optimal joint allocation) per round. The \textsc{monitor} transition freezes $\hat\tau_k$ on the first feasible success, and tasks never re-enter \textsc{search}, so no further structural regret accrues. Summing over tasks yields the first term.

\emph{(ii) Statistical monitoring and tail failures.} Conditioned on feasible execution, the per-pull reward sequence $\{X_{k,t} v_k\}$ is i.i.d.\ with support $\{0, v_k\}$. Crucially, the empirical mean $\hat\mu_k \leftarrow (1/N_k) \sum_{s \leq t : \text{feasible exec on } k} Y_{k,s}$ and pull counter $N_k$ are only updated during these feasible executions, ensuring the estimates remain unbiased by censored feedback. At any round $t$, the Hoeffding bound yields a confidence radius of $v_k \sqrt{c \log t / N_k}$ for a constant $c > 0$, cleanly recovered by our planning index. Applying the gap-dependent regret bound for combinatorial semi-bandits with multi-play feedback~\cite{combes2015combinatorial} bounds the statistical regret by $O(v_k^2 \log T / \Delta_k)$ per feasible task. The tail failure term $O(K \Delta_{\max})$ arises from the exponentially decaying failure tails of the confidence intervals~\cite{auer2002finite,chen2013combinatorial}.
\end{proof}

\paragraph{Implementation details.}
While Theorem~\ref{thm:regret} dictates a conservative failure budget that scales with $\log T$, in our experiments we use a fixed budget of $N_{\max} = 5$, which provides sufficient feasibility-resolution power for the environments we consider. For the exploration radius, we use $c = 2$, the standard UCB1 constant. The mean-reward estimate $\hat\mu_k$ is updated only on feasible executions; failures below $\hat\tau_k$ contribute only to the \textsc{search}-phase pruning counter. The per-round computational cost is $\mathcal{O}(KM)$ to exactly solve the integer knapsack via dynamic programming, plus $\mathcal{O}(K)$ for the UCB index updates.

\section{Decentralized Coordination}
\label{sec:dtac}

Unlike the centralized baseline, decentralized agents lack access to a global view of outcomes. They observe only their own rewards and must rely on communication to synchronize their beliefs. The primary challenge is not just statistical estimation, but \emph{alignment on feasibility}: if agents disagree on the current hypothesis $\hat{\tau}_k$, they may compute conflicting optimal plans, leading to mis-coordinated coalitions that fail to trigger tasks.

To address this, we introduce \textbf{D-TAC} (Decentralized Threshold-Activated Coordination). D-TAC instantiates a \emph{Virtual Coordinator} at each agent: a local copy of the C-TAC planner (Section~\ref{sec:ctac}) that operates on the agent's belief state. Because every agent runs the same deterministic planner under shared rules, identical belief states produce identical joint plans without explicit negotiation. Following the ``Public Agent'' paradigm~\cite{chakraborty2017coordinated}, agents share four pieces of structure prior to deployment: the C-TAC planner, a rank-based assignment rule, belief-state fusion rules, and common synchronization triggers. Under these shared rules, communication is required only to maintain belief-state consistency, not to coordinate actions. The contribution of D-TAC lies in the event-triggered protocol that schedules belief synchronization based on structural changes; the action-coordination mechanism is inherited from prior work. Because feasibility threshold estimates are non-decreasing under conservative max-fusion and bounded by $M$, structural disagreement is self-limiting and the team stabilizes on a shared hypothesis within a finite number of events (Proposition~\ref{prop:comm}).

\subsection{The Virtual Coordinator Architecture}
\label{sec:vc-arch}

Each agent $i$ maintains a local belief state $\mathcal{B}_i^t$ comprising, for each task $k$: a threshold lower bound $\hat{\tau}_k^{lo}$, a threshold upper bound $\hat{\tau}_k^{hi}$, a phase indicator $\phi_k \in \{\textsc{search}, \textsc{monitor}, \textsc{infeasible}\}$, and reward statistics $(\hat{\mu}_k, N_k)$. The agent also stores the most recent joint plan $\boldsymbol{c}^*$. The lower bound $\hat{\tau}_k^{lo}$ is initialized to $1$ and advances on Type~II events; the upper bound $\hat{\tau}_k^{hi}$ is initialized to $M$ (the maximum coalition size) and tightens to $\min(\hat{\tau}_k^{hi}, c_{k,t})$ on observing $Y_{k,t} > 0$. 
The planner uses $\hat{\tau}_k^{hi}$ as the coalition-size weight, since it represents the smallest empirically-confirmed feasible size. Each agent also stores $\hat{\tau}_k^{lo,\text{synced}}$ and $\phi_k^{\text{synced}}$, the values of $\hat{\tau}_k^{lo}$ and $\phi_k$ at the most recent sync, used to detect Type~I breakthroughs against the last globally-known hypothesis.

\paragraph{Deterministic Consensus.}
\label{sec:rank-based}
We use \emph{consensus} in the engineering sense of a shared assignment rule, not in the distributed-systems sense of a negotiated protocol: there is no message exchange in this step. Agents are indexed by unique IDs $i \in \{1, \ldots, M\}$. If the joint plan $\boldsymbol{c}^*$ allocates $c_k^*$ agents to task $k$, agents are assigned in ID order: the first $c_1^*$ to task $1$, the next $c_2^*$ to task $2$, and so on. Agents with index beyond $\sum_k c_k^*$ idle that round. Since this rule is deterministic and depends only on $\boldsymbol{c}^*$, any two agents holding identical plans produce identical assignments. Transient disagreement may occur between syncs; agents persist with the last agreed-upon plan (``sticky execution'') to avoid mis-coordination.

\subsection{Structure-Aware Communication Protocol}
\label{sec:comm-protocol}

D-TAC uses an event-triggered protocol with a low-frequency heartbeat backstop, optimizing for the number of synchronization rounds rather than payload size (consistent with Section~\ref{sec:observations}). Agents operate silently by default. A sync is triggered when one of three events occurs:

\begin{itemize}
\item \textbf{Type~I (Feasibility Breakthrough):} agent $i$ observes $Y_{k,t} > 0$ at $c_{k,t} < \hat{\tau}_k^{lo,\text{synced}}$ or at a task with $\phi_k^{\text{synced}} = \textsc{infeasible}$. 
\item \textbf{Type~II (Structural Pruning):} agent $i$ accumulates $N_{\max}$ consecutive informative failures (coalition $\geq \hat{\tau}_k^{lo}$ in phase \textsc{search}), advancing $\hat{\tau}_k^{lo}$ and changing the planner's coalition-size weight. 
\item \textbf{Periodic Heartbeat:} every $T_h$ rounds, to bound divergence in reward estimates. Performance is robust to $T_h$ over a wide range.
\end{itemize}

D-TAC begins with a warmup phase ($t = 1, \ldots, K$) in which all $M$ agents probe each task once; a single sync at $t = K{+}1$ pools the warmup observations.

When a sync fires, agents broadcast belief states and fuse peer states under three rules. Threshold bounds are fused conservatively: $\hat{\tau}_k^{lo} \leftarrow \max_j \hat{\tau}_{k,j}^{lo}$ (monotone non-decreasing) and $\hat{\tau}_k^{hi} \leftarrow \min_j \hat{\tau}_{k,j}^{hi}$ (tightening on any peer's success). Phases combine under the precedence $\textsc{monitor} \succ \textsc{infeasible} \succ \textsc{search}$, so a single peer with empirically confirmed feasibility overrides others' \textsc{search} or pessimistic \textsc{infeasible} claims. Reward estimates are fused as the count-weighted mean $\hat{\mu}_k \leftarrow (\sum_j N_{k,j} \hat{\mu}_{k,j}) / (\sum_j N_{k,j})$ over local observations accumulated since the last sync, yielding the same posterior as pooling raw samples. Fusion is idempotent: a repeated sync with no new observations leaves belief components unchanged.

\begin{algorithm}[t]
\caption{D-TAC (Agent $i$ View)}
\label{alg:dtac}
\begin{algorithmic}[1]
\REQUIRE Tasks $\mathcal{K}$, Agents $M$, Heartbeat $T_h$, Failure budget $N_{\max}$
\STATE \textbf{Init:} $\mathcal{B}_i \leftarrow$ priors; $\boldsymbol{c}^* \leftarrow \text{Planner}(\mathcal{B}_i)$
\STATE \textbf{Warmup:} for $t = 1, \ldots, K$, all $M$ agents probe task $t$; sync at $t = K{+}1$
\FOR{$t = K{+}1, \dots, T$}
    \STATE $\text{Sync} \leftarrow (t \bmod T_h = 0)$
    \IF{Type~I or Type~II event since last sync}
        \STATE $\text{Sync} \leftarrow \text{True}$
    \ENDIF
    \IF{Sync}
        \STATE Broadcast $\mathcal{B}_i$; receive $\{\mathcal{B}_j\}$
        \STATE Fuse: $\hat{\tau}_k^{lo} \leftarrow \max_j$, $\hat{\tau}_k^{hi} \leftarrow \min_j$, $\phi_k$ by precedence, $\hat{\mu}_k$ count-weighted
        \STATE $\hat{\tau}_k^{lo,\text{synced}} \leftarrow \hat{\tau}_k^{lo}$; $\phi_k^{\text{synced}} \leftarrow \phi_k$; $\boldsymbol{c}^* \leftarrow \text{Planner}(\mathcal{B}_i)$
    \ENDIF
    \STATE Execute $k \leftarrow \boldsymbol{c}^*[i]$; observe $Y_{k,t}$
    \IF{$Y_{k,t} > 0$}
        \STATE $\hat{\tau}_k^{hi} \leftarrow \min(\hat{\tau}_k^{hi}, c_{k,t})$
        \IF{$c_{k,t} < \hat{\tau}_k^{lo}$}
            \STATE $\hat{\tau}_k^{lo} \leftarrow 1$ \COMMENT{refutation}
        \ENDIF
        \STATE Update $\hat{\mu}_k, N_k$
        \IF{$\phi_k \in \{\textsc{search}, \textsc{infeasible}\}$}
            \STATE $\phi_k \leftarrow \textsc{monitor}$
        \ENDIF
        \STATE Flag Type~I if $c_{k,t} < \hat{\tau}_k^{lo,\text{synced}}$ or $\phi_k^{\text{synced}} = \textsc{infeasible}$
    \ELSIF{$\phi_k = \textsc{search}$ \AND $c_{k,t} \geq \hat{\tau}_k^{lo}$}
        \STATE $N_{\text{fail},k} \leftarrow N_{\text{fail},k} + 1$
        \IF{$N_{\text{fail},k} \geq N_{\max}$}
            \STATE Advance $\hat{\tau}_k^{lo}$; flag Type~II
        \ENDIF
    \ENDIF
\ENDFOR
\end{algorithmic}
\end{algorithm}

\paragraph{Convergence and Communication Bound.}
D-TAC is not designed to provide worst-case regret guarantees under decentralization in this work; a formal characterization requires bounding the transient phase of mismatched plans, the expected time to belief alignment, and the impact of stale estimates, which we defer to ongoing work. Instead, we establish a finite bound on the structural communication overhead.

\begin{proposition}[D-TAC Communication Complexity]
\label{prop:comm}
Under the D-TAC protocol, the total number of structural synchronization events (Type~I and Type~II triggers) over any horizon $T$ is at most $\mathcal{O}(KM)$ network-wide before structural consensus is established at sync events.
\end{proposition}

\begin{proof}[Proof sketch]
We bound the number of sync events triggered by structural changes 
team-wide. Type~II events advance $\hat{\tau}_k^{lo}$, whose team-synced 
value is monotone non-decreasing through $\{1, \ldots, M{+}1\}$ under 
max-fusion; at most $M$ such advances per task. Type~I events fire when 
successes refute the team's synced hypothesis; each such event tightens 
the team-synced upper bound, which is monotone non-increasing through 
the same range under min-fusion; at most $M$ such tightenings per task. 
Summing over $K$ tasks gives $\mathcal{O}(KM)$ structural sync events 
team-wide; beyond this, only the heartbeat fires.
\end{proof}

Once structural beliefs align (within $\mathcal{O}(KM)$ events) and all feasible tasks transition to \textsc{monitor}, D-TAC's planner produces the same plan that C-TAC would given the shared belief state. From this point, regret accrues only from statistical estimation and the bounded lag between heartbeat syncs, so D-TAC empirically approaches C-TAC's asymptotic performance, modulo a transient during the structural-learning phase. The total communication budget under D-TAC is therefore $\mathcal{O}(KM + T/T_h)$, with heartbeats dominating once structural consensus is reached.

\section{Experimental Evaluation}
\label{sec:experiments}

We evaluate TAC-MAB to quantify the cost of learning under censored feedback and to assess whether explicit coordination is necessary to escape feasibility-gated failure modes. Our experiments isolate \emph{structural difficulty} (identifying unknown coalition-size requirements) from standard statistical estimation error, and compare centralized and decentralized coordination under a shared observation model.

\paragraph{Experimental setup and baselines.}
All results are averaged over $N=40$ independent runs with different random seeds; shaded regions and error bars indicate $\pm 1$ standard error. Unless otherwise stated, methods are evaluated with horizon $T=10{,}000$, $M=5$ agents, and $K=10$ task types. The environment includes two infeasibility decoys ($\tau > M$), one full-team task ($\tau = M$), several mid-coordination tasks, and low-threshold distractors. D-TAC uses heartbeat period $T_h=50$ and failure budget $N_{\max}=5$. C-TAC uses an exact knapsack planner (Section~\ref{sec:ctac}), isolating learning and coordination effects from computational approximation.

We compare the following strategies:
\begin{itemize}
    \item \textbf{Oracle Allocation}: a full-information benchmark that knows all feasibility thresholds and success probabilities and selects the optimal allocation in every round.
    \item \textbf{Independent UCB}: a decentralized baseline in which agents independently estimate task values using UCB without communication or coordination.
    \item \textbf{C-TAC}: the centralized feasibility-aware strategy described in Section~\ref{sec:ctac} (Algorithm~\ref{alg:ctac}), which explicitly probes coalition sizes to resolve feasibility.
    \item \textbf{D-TAC}: the decentralized coordination protocol introduced in Section~\ref{sec:dtac} (Algorithm~\ref{alg:dtac}), which synchronizes agents only upon structural belief updates.
\end{itemize}
All methods receive the same local observations as defined in Section~\ref{sec:observations}; only coordination and communication differ. We omit direct comparison against decentralized cooperative-MAB baselines such as Landgren-CoopUCB~\cite{landgren2021distributed} and DDUCB~\cite{martinez2018decentralized}: faithful reproduction requires implementation choices (e.g., UCB exploration constants, fusion conventions) that could not be calibrated against authors' reference implementations. We defer this comparison to follow-up work.

\paragraph{Censored feedback breaks independent exploration.}
We first consider an environment containing one high-value cooperative task that requires full-team coordination and several low-threshold distractor tasks that provide immediate but smaller rewards. This setting isolates the censored-feedback failure mode: without explicit coordination, independent agents are unlikely to activate the cooperative task and instead converge to suboptimal alternatives.

Figure~\ref{fig:regret_time} reports cumulative regret over time. Independent UCB exhibits sustained linear regret, reflecting its inability to reliably form sufficiently large coalitions under censored feedback. Failures below the feasibility threshold provide no informative signal, so independent exploration does not guide agents toward larger coalitions. C-TAC (centralized) and D-TAC (decentralized) incur an initial coordination cost during early structure learning; once feasibility is resolved, regret growth slows substantially. The cost of learning feasibility is incurred early and does not scale with the horizon.

\paragraph{Communication efficiency.}
At $T{=}10{,}000$, D-TAC achieves cumulative regret of $13{,}700$ using $4{,}303$ messages, versus C-TAC's $4{,}500$ regret at $100{,}000$ messages---a $23{\times}$ reduction in communication while remaining within the same order of magnitude in regret. D-TAC's message budget is amortized over structural events: the periodic heartbeat accounts for most messages once feasibility is resolved, while Type~I and Type~II structural triggers fire only during the learning phase.

\begin{figure}[t]
    \centering
    \includegraphics[width=\columnwidth]{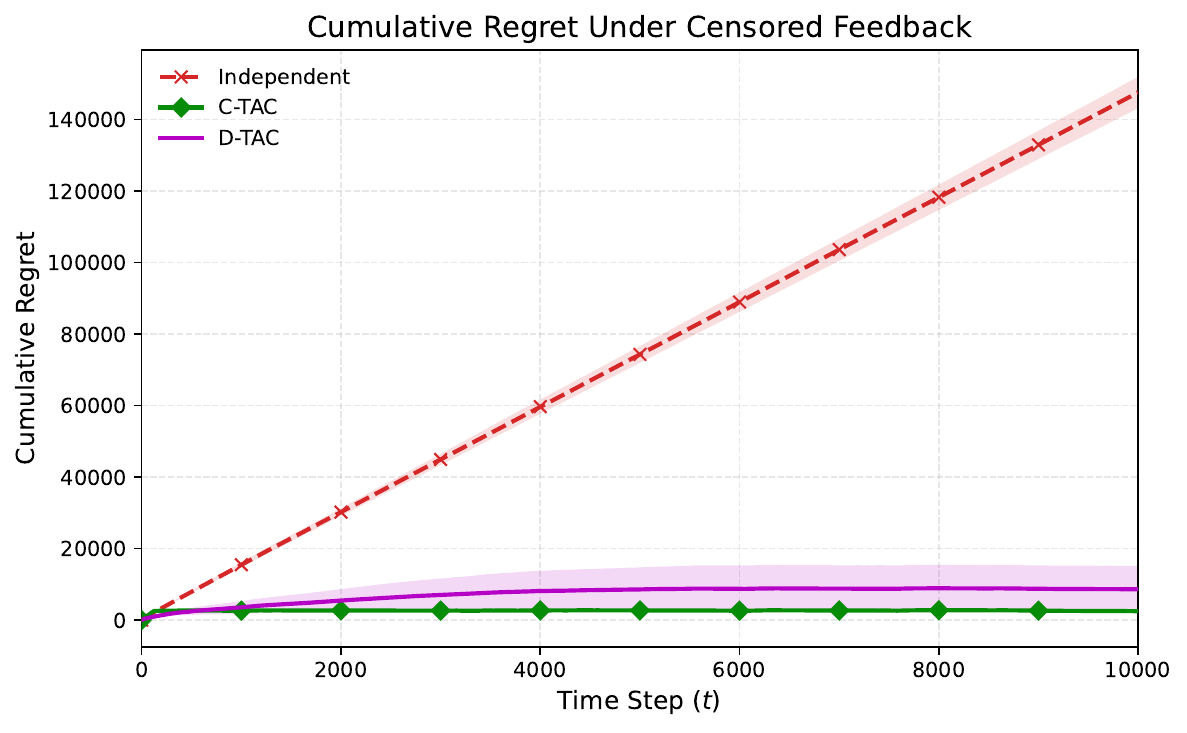}
    \caption{\textbf{Cumulative regret under censored feedback.}
    Cumulative regret $R(t)$ over $T=10{,}000$ steps, averaged over $N=40$ independent runs.
    Independent UCB exhibits sustained linear regret due to censored feedback.
    C-TAC and D-TAC incur an early coordination cost to resolve feasibility, after which regret growth slows.
    Shaded regions denote $\pm 1$ standard error.}
    \label{fig:regret_time}
\end{figure}

\paragraph{Structural hardness induced by feasibility thresholds.}
We next examine how learning difficulty scales with the magnitude of coordination requirements. We vary the maximum feasibility threshold $\tau_{\max}$ across environments while keeping the number of agents, tasks, rewards, and success probabilities fixed. As $\tau_{\max}$ increases, informative feedback becomes increasingly unlikely under uncoordinated exploration.

Figure~\ref{fig:tau_sweep} reports the final cumulative regret $R(T)$ as a function of $\tau_{\max}$. Independent UCB suffers rapidly increasing regret as coordination requirements grow, reflecting the vanishing probability of independently activating threshold-gated tasks. C-TAC maintains consistently low regret across all threshold levels by explicitly probing feasibility. D-TAC degrades gracefully and approaches centralized performance, demonstrating that event-triggered coordination preserves feasibility learning while substantially reducing communication.

\begin{figure}[t]
    \centering
    \includegraphics[width=\columnwidth]{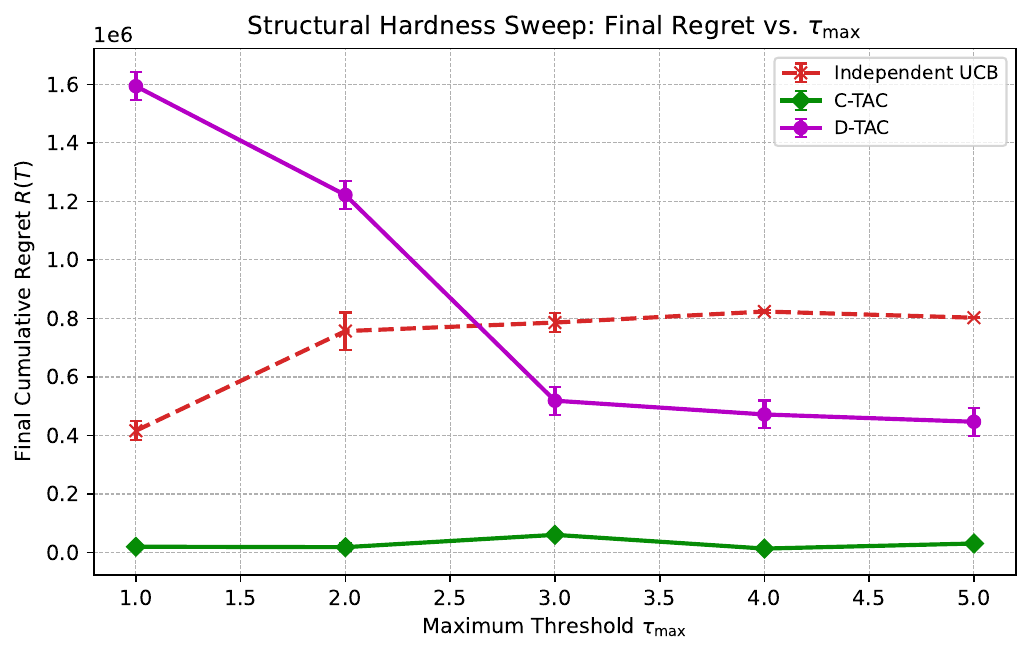}
    \caption{\textbf{Structural hardness versus regret.}
    Final cumulative regret $R(T)$ as a function of the maximum feasibility threshold $\tau_{\max}$.
    Independent UCB exhibits increasing regret as coordination requirements grow.
    C-TAC maintains low regret across all threshold levels.
    D-TAC approaches centralized performance while substantially reducing communication.
    Results are averaged over $N=40$ runs with horizon $T=10{,}000$; error bars denote $\pm 1$ standard error.}
    \label{fig:tau_sweep}
\end{figure}

\paragraph{Limitations.}
\label{sec:limitations}
Our experiments focus on stationary feasibility thresholds and synchronous rounds. In settings with non-stationary coordination requirements, delayed communication, or adversarial failures, additional mechanisms may be required to maintain feasibility alignment. While D-TAC substantially reduces communication in practice, we do not provide worst-case decentralization regret guarantees; characterizing the minimal communication required for feasibility learning under adversarial or dynamic conditions remains an open direction.

\section{Related Work}
\label{sec:related}

Our work intersects with several primary domains in sequential decision-making, each with distinct limitations when applied to decentralized, threshold-gated coordination.

\paragraph{Multi-Agent Multi-Armed Bandits (MAMAB).}
Recent advances in MAMAB explore communication-constrained coordination~\cite{landgren2021distributed,martinez2018decentralized,chang2022online,chakraborty2017coordinated,agarwal2022multi}, including event-triggered and on-demand protocols~\cite{chen2023demand}, gossip-based information sharing~\cite{chawla2020gossiping}, heterogeneous and collision-aware decentralized learning~\cite{kalathil2014decentralized,magesh2021decentralized}, and fully-decentralized cooperation without communication~\cite{chang2023optimal}. These formulations assume rewards are triggered by individual agent actions or independent arm pulls. They fail in environments where task feasibility is gated by coalition size, as independent exploration yields fully censored feedback, preventing agents from learning task values. Related work in ad hoc teamwork~\cite{barrett2014communicating} addresses coordination among agents with unknown communication protocols but focuses on uncertainty over teammate behavior rather than environmental coordination requirements.

\paragraph{Combinatorial and Knapsack Bandits.}
Centralized combinatorial multi-armed bandits (CMAB) and knapsack-based bandits study how to allocate limited resources to maximize reward over time~\cite{chen2013combinatorial,combes2015combinatorial,tran2012knapsack,das2022budgeted}. These models typically assume that resource costs or activation conditions are known in advance. In our setting, the feasibility threshold is a latent parameter that must be learned online, imposing an additional structural learning cost that standard CMAB frameworks do not address.

\paragraph{Censored Feedback and Resource Allocation.}
Bandit models with threshold-based or censored feedback have recently gained traction. Abernethy et al.~\shortcite{abernethy2016threshold} study settings where rewards are observed only if expected values exceed a known threshold; Zhang et al.~\shortcite{zhang2024learning} extend this to actively learning thresholds with latent values under single-agent censored feedback. Most closely related, Verma et al.~\shortcite{verma2019censored} introduce censored semi-bandits. This literature is restricted to centralized settings allocating a continuous, divisible resource budget; it does not address the discrete, integer-based coordination of distinct agents, nor does it provide mechanisms for decentralized agents to align on unknown thresholds under communication constraints.

\paragraph{Our Contributions.}
TAC-MAB and D-TAC address these gaps. We model task requirements as integer coalition thresholds and prove that the proposed centralized algorithm (C-TAC) achieves $O(\log T)$ cumulative regret, decomposed into a structural-search term capturing the cost of resolving feasibility under censored feedback and a statistical-monitoring term for value estimation (Theorem~\ref{thm:regret}). We then introduce a decentralized event-triggered protocol (D-TAC) that maintains feasibility alignment under conservative max-fusion and empirically achieves a $23{\times}$ reduction in communication relative to the centralized baseline (Section~\ref{sec:experiments}).

\section{Conclusion}

We introduced the Threshold-Activated Cooperative Multi-Armed Bandit (TAC-MAB), a framework for cooperative tasks whose rewards are gated by unknown coalition-size requirements and are fully censored below feasibility. Independent exploration fails structurally in this setting: agents receive no informative feedback below threshold and cannot distinguish infeasibility from stochastic failure.

Decomposing regret into feasibility learning and value estimation, we proved that the centralized algorithm achieves $O(\log T)$ cumulative regret under the identifiability condition $p_k \geq p_{\min} > 0$ for feasible tasks (Theorem~\ref{thm:regret}), with feasibility resolving in $O(\log T / p_{\min})$ probes per task. Empirically, C-TAC resolves feasibility via targeted coalition-size probing, while D-TAC approximates this behavior through event-triggered synchronization---achieving a $23{\times}$ reduction in communication relative to the centralized baseline while preserving feasibility alignment under conservative max-fusion.

Extending D-TAC to lossy or adversarial communication remains open; we are pursuing a formal regret characterization under intermittent communication in ongoing work.

\appendix
\section*{Acknowledgements}

The conclusions and opinions expressed in this research paper are those of the authors and do not necessarily reflect the official policy or position of the U.S. Government or Department of Defense.

\bibliographystyle{named}
\bibliography{ijcai26}

\end{document}